\documentclass[draftclsnofoot,onecolumn,12pt]{IEEEtran}

\usepackage{cite}
\usepackage{graphicx}
\usepackage{psfrag}
\usepackage{subfigure}
\usepackage{amsmath}
\usepackage{amssymb}
\usepackage{epsfig}
\usepackage{color}
\usepackage{epsf}
\usepackage{algorithmic}
\usepackage{algorithm}
\usepackage{setspace}

\newtheorem{theorem}{Theorem}

\newtheorem{Coroll}{Corollary}

\newtheorem{Lemm}{Lemma}

\def\bb0{{\mathbb{0}}}

\def\bb{{\mathbf{b}}}

\def\bee{{\mathbf{e}}}
\def\bff{{\mathbf{f}}}
\def\bg{{\mathbf{g}}}
\def\bh{{\mathbf{h}}}

\def\bv{{\mathbf{v}}}
\def\bw{{\mathbf{w}}}

\def\b0{{\mathbf{0}}}

\def\bbC{{\mathbb{C}}}

\def\bbE{{\mathbb{E}}}

\def\bbP{{\mathbb{P}}}

\def\cC{\mathcal{C}}

\def\cF{\mathcal{F}}

\def\cI{\mathcal{I}}
\def\cJ{\mathcal{J}}

\def\cL{\mathcal{L}}

\def\cN{\mathcal{N}}

\def\cQ{\mathcal{Q}}

\def\sfb{{\mathsf{b}}}

\def\sft{{\mathsf{t}}}

\def\sf0{{\mathsf{0}}}

\def\rmB{\mathrm{B}}

\def\rmD{\mathrm{D}}

\def\rmQ{\mathrm{Q}}
\def\rmR{\mathrm{R}}

\def\rmb{{\mathrm{b}}}

\def\rmd{{\mathrm{d}}}

\def\rmt{{\mathrm{t}}}

\def\rm0{{\mathrm{0}}}

\DeclareMathOperator*{\argmax}{arg\,max}

\begin{document}
%

\title{Limited Feedback Over Temporally Correlated Channels for the Downlink of a Femtocell Network}
\author{Salam~Akoum, Marios~Kountouris and~Robert~W.~Heath,~Jr.
\thanks{Copyright (c) 2010 IEEE. Personal use of this material is permitted. However, permission to use this material for any other purposes must be obtained from the IEEE by sending a request to pubs-permissions@ieee.org.}
\thanks{S. Akoum and R. W. Heath, Jr. are with the Wireless Networking and Communications Group, Department of Electrical and Computer Engineering, The University of Texas at Austin, 1 University Station C0803, Austin, TX, 78712-0240, USA e-mail: \{salam.akoum, rheath\}@mail.utexas.edu. This work was partly supported by the Semiconductor Research Company (SRC) Global Research Consortium (GRC) task ID 1836.024. It was partly supported by the ECOSCELLS project. }
\thanks{Marios Kountouris is with the department of Telecommunications,  SUP{\'E}LEC, Plateau de Moulon, 3 rue Joliot-Curie, 91192 Gif Sur Yvette, France. Email: marios.kountouris@supelec.fr}}

\maketitle

%
%
%

%
\begin{abstract}
Heterogeneous networks are a flexible deployment model that rely on low power nodes to improve the user broadband experience in a cost effective manner. Femtocells are an integral part of heterogeneous networks, whose main purpose is to improve the indoor capacity. When restricting access to home users, femtocells cause a substantial interference problem that cannot be mitigated through coordination with the macrocell base station. In this paper, we analyze multiple antenna communication on the downlink of a macrocell network, with femtocell overlay. We evaluate the feasibility of limited feedback beamforming given delay on the feedback channel, quantization error and uncoordinated interference from the femtocells. We model the femtocell spatial distribution as a Poisson point process and the temporal correlation of the channel according to a Gauss-Markov model. We derive the probability of outage at the macrocell users as a function of the temporal correlation, the femtocell density, and the feedback rate. We propose rate backoff to maximize the average achievable rate in the network. Simulation results show that limited feedback beamforming is a viable solution for femtocell networks despite the CSI inaccuracy and the interference. They illustrate how properly designed rate backoff improves the achievable rate of the macrocell system.

\end{abstract}
\newpage
%
\section{Introduction}
The best way to increase the capacity of cellular systems is by decreasing the cell size \cite{Mulder2010, Claussen2008, Chandrasekhar2008}. Unfortunately, relying on cell splitting in traditional network deployment approaches is complex in high density environments. Heterogeneous networks provide a more flexible deployment model that relies on base stations of diverse properties to improve the user broadband experience in a ubiquitous and cost effective manner \cite{Ghosh2010}. They consist of macro base stations overlaid with low power nodes such as pico base stations, distributed antennas, relays, and (of primary interest in this paper) femtocells. Unlike distributed antennas and relays, femtocells are user deployed and unplanned by the network operators. They are short range, operate on licensed spectrum and are usually restricted to serve home or enterprise users only. Unfortunately, restricting access to the femtocells to subscribed home users causes significant interference between adjacent femtocells, as well as cross tier interference between the femtocells and the macrocells \cite{Chandrasekhar2008, delaroche2010}. Moreover, interference management in femtocell networks does not conform with traditional cellular networks \cite{Claussen2010}, \cite{Yavuz2009}, as it cannot rely on coordination between the macrocells and the femtocells and needs to be decentralized.

As interest in multiple input multiple output (MIMO) communication has grown, upcoming cellular standards have embraced multiple antennas on the base stations and the mobile users to increase data rates and improve performance of the radio link \cite{juho2009}. Multiple antennas are being considered at the femtocells  for coverage improvement and interference management \cite{Chandrasekhar2008, Chandrasekhar2009}. One of the most practical transmission strategies for realizing the benefits of MIMO systems is limited feedback linear precoding \cite{OverviewLFWC}. Limited feedback suffers, however, from a CSI inaccuracy due to feedback delay and quantization error \cite{Caire2007, Nguyen2004}. Its performance further degrades in cellular systems, when the interference is not coordinated, and the transmitter does not have instantaneous knowledge about the interference at the receiver \cite{Akoum2010}. In this case, the transmitted rate does not take into account the accurate interference information, resulting in an increased probability of outage. In this paper, we consider limited feedback beamforming on the downlink of a heterogeneous network with private access femtocells, and we evaluate its performance, given the cross tier interference in the network, and the feedback delay.

\subsection{Prior Work}

Prior research on femtocell networks has focused on decentralized and self-optimization strategies for interference management as well as femtocell access control \cite{Yeh2008, Claussen2007, Ho2007, Claussen2008, Xia2010, LopezPerez2009, Simeone2010p}. The feasibility of deploying femtocells in the same frequency bands as the existing macrocell, and methods for femtocell power auto-configuration and public access were investigated in \cite{Claussen2007, Claussen2008}. In \cite{Ho2007}, the authors analyzed the effect of the additional handovers caused by femtocells on the dropped call rate of the macrocell users.  Interference avoidance techniques such as self-optimization and dynamic frequency planning were suggested in \cite{LopezPerez2009} to improve the capacity and decrease interference for an OFDMA femtocell. The performance of the femtocell access control strategy was investigated in \cite{Simeone2010p}, with multi-cell and single-cell processing. Femtocell access strategy was discussed in  \cite{Xia2010}, where it was shown that for OFDMA femtocells, the femtocell access control should be adapted to the cellular user density.  

Thus far there has been limited work on the application of multiple antennas to heterogeneous networks with femtocells. The benefits of using multiple antennas at the femtocells and the macro base stations, in terms of coverage and spatial reuse, were evaluated in \cite{Chandrasekhar2009}. The analysis, based on zero forcing precoding, assumed full CSI at the transmitter,  and suggested that  single user transmission outperforms multi user transmission on the downlink of heterogeneous networks.  Precoding techniques at the femtocells, based on having perfect or quantized CSI at the transmitter, were investigated in \cite{Jiang2010}, \cite{Husso2010} to reduce the interference, perform local coordination, or analyze the best transmission strategy for the femtocells. The work considered two adjacent femtocell base stations only, and failed to capture the randomness in the deployment of the femtocells, as well as the feedback delay. Furthermore, evaluating the effect of cross tier interference on the performance of the cellular network borrows from the literature on transmission capacity for ad hoc networks with multiple antennas \cite{Hunter2008, Vaze2009, Huang2010, Kountouris2010, Jindal2011}, where the outage probability and transmission capacity for multiple antenna diversity and spatial multiplexing techniques were derived. An expression for the asymptotic spectral efficiency in an interference limited regime was derived in \cite{Govindasamy2010}. The analysis considered interference from spatially distributed nodes with CSI at the transmitter, and concluded that CSI increases the spectral efficiency. 

To study the feasibility of limited feedback beamforming on the downlink of a heterogeneous network,  the analysis should take into account the effect of feedback delay and quantization error. Feedback delay arises from signal processing algorithms, propagation, and channel access protocols. The achievable rate and bit error rate performances of MIMO systems with feedback delay have been investigated in  \cite{FiniteRateRao, MarkovKroenecker, Nguyen2004 , Zhang2009, Caire2007, LFTCJ}. The delay has been found to reduce the achievable throughput \cite{LFTCJ, MarkovKroenecker}, and to cause interference between spatial data streams \cite{Nguyen2004}. Delay has been shown in \cite{Zhang2009} to incur a significant performance loss on multi user MIMO systems, due to residual inter user interference, and in \cite{Caire2006} to limit the performance of joint precoding and scheduling for the MIMO broadcast channel. The joint effect of delay and other cell interference assuming one strong interferer has been shown in \cite{Akoum2010} to increase the decay in the achievable rate using a Markov chain model. 
In this paper, to investigate the applicability of limited feedback methods, we consider the joint effect of feedback delay and cross tier interference. This has not been investigated in the literature for randomly distributed interferers. 

\subsection{Contributions}
In this paper, we consider limited feedback beamforming on the downlink of a heterogeneous network, with a femtocell overlay. We compute the probability of successful reception and the achievable rate for a macrocell user in the presence of cross tier interference from the femtocells. We account for the joint effect of feedback delay, quantization error and cross tier interference as follows. The propagation information that reaches the base station, being quantized and delayed, causes a CSI mismatch between the transmitter and the receiver. The mismatch translates into packet outage when the transmitted rate exceeds the instantaneous mutual information of the channel \cite{LauV, Aggarwal, Akoum2010}. When the interference is not coordinated, the transmitter is oblivious to the instantaneous interference at the receiver, and modulates its information at a rate corresponding to the estimated signal to interference ratio ($\mathsf{SIR}$) rather than the actual $\mathsf{SIR}$, hence increasing the probability of outage. 

We consider a narrowband channel model with a Gauss Markov temporal correlation and a distance dependent path loss model. Since we compute the average rates as a function of the distance and we average over all locations inside the macrocell, we neglect large scale fading. The base stations have multiple antennas while the subscribed users have a single antenna. Extending the results to multiple receive antennas and using precoded spatial multiplexing is a subject of future investigation. We assume that the femtocells are distributed according to a homogeneous Poisson point process (PPP). We derive an expression for the probability of successful reception of limited feedback beamforming as a function of feedback delay, quantization size, the distribution and the density of the femtocell interferers in the network. The general mathematical framework is based on properties of the Chi-squared distributions and the homogeneous PPP. We use the probability of successful reception to compute a closed form expression for the maximum number of femtocells that can be deployed inside the macrocell as a function of the distance between the mobile user and the macro base station. 

We show that limited feedback beamforming outperforms no beamforming, in terms of average achievable rate. Thus we conclude that limited feedback beamforming is viable for transmission on the downlink of a femtocell network. To evaluate the achievable rate, we use the goodput metric \cite{LauV, Aggarwal, Akoum2010}. Goodput is defined as the number of bits successfully transmitted to the receiver per unit of time. It takes into account the packet outage caused by the CSI mismatch between the transmitter and the receiver. When the probability of outage goes to zero, it is equal to the throughput. To combat the effect of packet outage, we implement a rate adaptation strategy based on rate backoff. To maximize the average goodput achieved, the transmitter, instead of transmitting at a rate corresponding to the estimated $\mathsf{SIR}$, backs off the estimated $\mathsf{SIR}$ by a multiplicative factor, $\beta\mathsf{SIR}$, and modulates its signal based on the backoff  $\mathsf{SIR}$.  The backoff factor $\beta$ is chosen to maximize the average achievable goodput. Rate backoff is a suboptimal adaptive transmission strategy used in \cite{Vakili2006, Zheng2008a}, to overcome channel estimation errors,  in the context of centralized and distributed opportunistic scheduling for single input single output channels. In our analysis, we derive the achievable rate with rate backoff when the feedback delay is the main performance bottleneck, as well as when the mobile user is subject to cross tier interference. We show that rate backoff recovers the throughput scaling with $\mathsf{SNR}$, and increases the average achievable rate, without requiring any coordination between the macrocell and the femtocells transmissions.

\subsection{Organization} The rest of the paper is organized as follows.  In Section \ref{sec:sysModel},
we describe the limited feedback cellular network considered, and we review the assumptions used throughout the paper.  Section \ref{sec:SysGoodput} introduces the
 system goodput as a function of the cross tier interference and the feedback delay. In Section \ref{sec:ProbSuccess}, we derive a lower bound on the probability of successful reception at the mobile user as a function of the femtocell density, the distance from the base station and the feedback delay. Section \ref{sec:Backoff} presents the rate backoff strategy suggested and derives the optimal backoff factor to mitigate the effect of  interference and maximize the achievable goodput.  Section
 \ref{sec:Simul} presents numerical results that illustrate the performance of the limited feedback system and present the achievable sum rate using rate backoff at the transmitter. This is followed by concluding remarks in Section \ref{sec:Conclusion}.

\subsection{Notation} Bold lowercase letters $\mathbf{a}$ are used to denote column vectors, bold uppercase letters $\mathbf{A}$ are used to denote matrices, non bold letters $a$ are used to denote scalar values, and caligraphic letters $\mathcal{A}$ are used to denote sets or functions of sets. Using this notion, $|a|$ is the magnitude of a scalar, $\|\mathbf{a}\|$ is the vector 2-norm, $\mathbf{A}^*$ is the conjugate transpose, $\mathbf{A}^T$ is the matrix transpose, $[\mathbf{A}]_{lm}$ is the scalar entry of $\mathbf{A}$ in the $\ell^{th}$ row and the $k^{th}$ column. We use $\mathbb{E}$ to denote expectation, $\mathrm{a}^{\mathrm{t}}$ to denote the metric $\mathrm{a}$ evaluated at the transmitter, and $\bar{a}$ to denote the average of a. 
\section{System Model and Assumptions}\label{sec:sysModel}
Consider a central macrocell $C_0$ overlayed with a tier of uncoordinated femtocells. The heterogeneous network considered is depicted in Figure~\ref{fig:cellsysfemto}. The macro base station $B_0$ is located at the center of the macrocell, and is equipped with $N_b$ antennas. Each femtocell is equipped with $N_f$ antennas. The user terminals in the macrocell and the femtocells have a single receive antenna.  We assume that each base station serves one active user at a time, using intra-cell orthogonal access. Both the macrocell and the femtocells employ a limited feedback beamforming system \cite{Love03, Jindal2006}. The user terminal estimates the downlink CSI sequence using pilot symbols sent by its serving base station. We assume that the receiver estimates the channel perfectly. The CSI is quantized by means of a codebook and the quantization index is sent to the base station via a limited feedback channel. Since the macro base station and the femtocells are not synchronized and are assumed not to coordinate their transmissions, the quantization index at the user terminal only takes into account the downlink channel from its base station. The feedback channel is assumed to be error free, with a fixed feedback delay $\rmd$. 

We use a stochastic geometry framework to model the distribution of the femtocells. The femtocells are assumed to be distributed according to a homogeneous spatial Poisson point process  $\Pi_f$ with a fixed average density of $\lambda_f$ femtocells per meters squared. The average number of actively transmitting femtocells is hence $No_f = \lambda_f C$, where $C = \pi R_c^2$ is the area of the macrocell, $R_c$ is the radius of the macrocell. The PPP model, in addition to being analytically tractable, is justifiable by arguing that the femtocell transmissions are uncoordinated. It is the worst case among all stationary point processes due to its maximum entropy property. We assume that cross tier interference is the main source of interference. We treat the aggregate co-channel interference from neighboring macro base stations as noise, in the interest of evaluating the effect of the interference from the femtocells on the downlink of the macrocell. We focus on the transmission on the downlink of the macrocell to evaluate the feasibility of limited feedback beamforming in this scenario. We assume that the femtocell user is well shielded from interference on the downlink because of the indoor penetration loss, and the high received signal strength from its serving femtocell, as compared to the macrocell and the adjacent femtocells. 

The received signal at the cellular mobile user $M_0$ is
\begin{eqnarray}
 y_0[n] &=& \sqrt{\rho_m}\;\rmD^{\frac{-\alpha_m}{2}}\bh^{*}_0[n]\bff_0[n-\rmd]s_0[n] + \sqrt{\rho_f}\sum_{F_i \in \Pi_f}{|\rmD_i|^{\frac{-\alpha_{f}}{2}} \bg^{*}_i[n]\bw_i[n]r_i[n]} + v_0[n],
\end{eqnarray}
where $\bh_0[n] \in \bbC^{N_b \times 1}$ is the downlink channel between $M_0$ and $B_0$, with zero mean, unit variance, i.i.d complex Gaussian entries.  The vector $\bg_i[n] \in \bbC^{N_f \times 1}$ is the downlink interference channel from the femtocell $F_i$ to $M_0$. The vector $\bff_0[n-\rmd] \in \bbC^{N_b\times 1}$ is the transmit beamforming vector at $B_0$ for $M_0$.  The vector $\bw_i[n] \in \bbC^{N_f \times 1}$ is the transmit beamforming vector used by femtocell $F_i$ to maximize the received power at its user terminal, for the transmitted signal $r_i$. The signals $s_0$ and $r_i$ are such 
that $\bbE[|s_0|^2] \leq P_t$, $\bbE[|r_i|^2] \leq P_{t_i}$, where $P_t$  and $P_{t_i}$ are the maximum transmit power at the base station and the femtocells, respectively.  $\rho_m$ and $\rho_{f}$ are the effect at the receiver of the transmit power, the carrier frequency and the wall partition loss, from the pathloss model. $\rmD$ is the distance between $B_0$ and its mobile user $M_0$, and $\rmD_i$ is the distance between $F_i$ and $M_0$. $\alpha_m > 2$ and $\alpha_f > 2$ denote, respectively, the pathloss exponent of the outdoor channel between $B_0$ and $M_0$, and the indoor to outdoor channel between $F_i$ and $M_0$. $v_0[n] \sim \cC\cN(0,N_0)$ is the additive white Gaussian noise  at $M_0$ with variance $N_0$.  

\section{The Goodput Model with Limited Feedback}\label{sec:SysGoodput}
In what follows, we derive the average goodput of limited feedback beamforming on the downlink of the macrocell network to study the viability of limited feedback in the heterogeneous network scenario. We then propose a rate backoff strategy to maximize the average achievable rate of the system.

The macro cell user terminal $M_0$ quantizes its channel ${\bh_0[n]}$ by means of a unit norm vector codebook $\cF = \{\bff_1, \bff_2,\cdots, \bff_N\}$, where $N = 2^\rmB$ is the size of the codebook. The quantizer function $\cQ$ chooses the beamforming vector $\bff_\ell$ that maximizes the $\mathsf{SNR}$ at $M_0$ such that
\begin{equation}
\bff_\ell[n] = \cQ\{{\bh}_0[n]\} = \argmax_{\bff_k \in \cF}{|{\bh}^*_0[n]\bff_k|^2}, \quad 1 \leq k \leq N.
\end{equation}
\normalsize
The channel $\bh_0[n]$ is then mapped to the index $I_n = \ell$ which is fed back to $B_0$ using B bits. 

The quantization process at $M_0$ does not take into account the cross tier interference from the femtocells, as the macro base station and the femtocells are not synchronized, due to the unplanned deployment of the femtocells. 
$M_0$ has no knowledge of the effective channels $\bg_i\bw_i$ from $F_i \in \Pi_f$.  
To modulate its signals, the base station estimates the signal to interference noise ratio ($\mathsf{SINR}$) based on the delayed and quantized CSI as well as the statistical knowledge of the interference. This estimated $\mathsf{SINR}$ might be different from the actual $\mathsf{SINR}$ estimated at the receiver. When the transmitted rate exceeds the supported rate of the channel, a rate outage or a packet outage occurs. To account for the rate outage, we use the system goodput as the performance metric. The goodput gives the amount of information successfully received (without retransmission) at the user terminal. It provides a model to compute the effect of outage on the achievable rate, and permits the design of transmission techniques to combat the outage and consequently decrease the number of retransmissions needed. Goodput is defined as
\begin{equation}
\Lambda[n] = \rmR^\rmt[n - \rmd] \cI\left(\rmR^\rmt[n - \rmd] \leq \rmR[n] \right),
\end{equation}
where $\cI(A)$ is the indicator function, which evaluates to 1 if the event $A$ is true, and $0$ otherwise. 

While hybrid automatic repeat request (HARQ) has been recently used in cellular systems to combat the CSI mismatch problem, we assume in this paper that HARQ is not present. In systems employing HARQ, the signals received in error are retransmitted to maximize the probability of successful reception at the user terminal. This increases the packet delay and increases the receiver complexity. Moreover, the maximum number of HARQ retransmissions is fixed, and the long term goodput achieved with HARQ depends on the number of successfully received packets after the retransmissions \cite{Wu2010, Chen2010}. 

The rate supported by the channel $\rmR[n, \rmd]$, for $M_0$, assuming Gaussian distributed transmit symbols $s_0[n]$, is written as a function of the signal to interference ratio $\mathsf{SIR}[n,\rmd]$, $$\rmR[n,\rmd] = \log_2\left(1 + \mathsf{SIR}[n,\rmd] \right),$$ where 
$\mathsf{SIR}[n,\rmd] = \frac{\rho_m\;\rmD^{-\alpha_m}|\bh^{*}_0[n]\bff_0[n-\rmd]|^2}{\rho_f\;\sum_{F_i\in\Pi_f}{|\rmD_i|^{-\alpha_{f}}|\bg^{*}_i[n]\bw_i[n]|^2}} 
=\frac{|\bh^{*}_0[n]\bff_0[n-\rmd]|^2}{\rmQ_{\rmD}I_f}.$
\normalsize

We assume that the system is interference limited, and that the received signal power is much higher than the noise power.
We thus ignore the thermal noise in the computations in this paper and mainly focus on the effect of cross tier interference on the achievable rate \cite{Chandrasekhar2009}. Furthermore, the expressions using $\mathsf{SIR}[n,\rmd]$ are more tractable analytically \cite{Hunter2008}. 

The interference function $I_f = \sum_{F_i\in\Pi_f}{|\rmD_i|^{-\alpha_{f}}|\bg^{*}_i[n]\bw_i[n]|^2}$
 is a shot noise process. The shot noise model is considered an efficient statistical model for predicting the fluctuations of the interference field, in the same way as the Rayleigh fading model is widely used to predict the fluctuations in the multipath channel,  \cite{Baccelli2009}. The pathloss ratio $\rmQ_\rmD = \frac{\rho_f}{\rho_m\rmD^{-\alpha_m}}$ is a function of the distance $\rmD$, and the ratio of  $\rho_f$ and $\rho_m$.  The transmitted rate is written in terms of  $\mathsf{SIR}^\sft[n,\rmd]$, the estimated $\mathsf{SIR}$ at $B_0$ as, 
\begin{equation}
\rmR^\rmt[n- \rmd]  =  \log_2\left(1 + \mathsf{SIR}^\sft[n,\rmd] \right),
\end{equation}
\normalsize
where 
$\mathsf{SIR}^\sft[n,\rmd] = \bar{\rho}|\bh^*[n - \rmd]\bff_0[n - \rmd]|^2$.

We assume that the transmitter has an estimate of the average signal to interference power ratio $\qquad\qquad\qquad \bar{\rho} = \frac{1}{\bbE\left[ \frac{\rho_f}{\rho_m\rmD^{-\alpha_m}}I_f\right]}= \frac{1}{\bbE\left[Q_\rmD I_f\right]}$.

The average goodput is expressed in terms of $\rmR^\sft[n-\rmd]$ and $\rmR[n,\rmd]$ as
\begin{eqnarray}\label{eqn:ergodic_goodput}
\nonumber\bar{\Lambda} &=& \bbE\left[ \rmR^\rmt[n-\rmd]\cI\left(\rmR^\sft[n-\rmd] \leq \rmR[n,\rmd]\right)\right]
\nonumber= \bbE\left[ \rmR^\rmt[n-\rmd]\bbP\left[\rmR^\sft[n-\rmd] \leq \rmR[n,\rmd] \;\large|\; \rmR^\sft[n-\rmd]  \right]\right]\\
&=& \int_{0}^{\infty}{\rmR^\rmt(\Upsilon) \bbP\left[\mathsf{SIR} \geq \Upsilon\;\large|\; \Upsilon \right] f_\Upsilon(\Upsilon)d\Upsilon}.
\end{eqnarray}
The integration is taken over $\mathsf{SIR}^\sft[n,\rmd]$, and $f_\Upsilon(\Upsilon)$ denotes the probability density function (pdf) of $\mathsf{SIR}^\sft[n,\rmd]$.  To evaluate $\bar{\Lambda}$, we find a closed form expression for the probability of successful reception $\bbP\left[\mathrm{SIR} \geq \Upsilon\right]$, as a function of the delay, the femtocell density and the distance between $M_0$ and $B_0$.

\section{The Probability of Successful Reception}\label{sec:ProbSuccess}
The probability of successful reception is evaluated as the complementary cumulative distribution function (CCDF) of the desired channel power $|\bh^{*}_0[n]\bff_0[n-\rmd]|^2$ given the interference shot noise $I_{f}$
 \small
\begin{eqnarray}
\bbP\left[\mathsf{SIR}[n,\rmd] \geq \Upsilon \right] &=& \bbP\left[\frac{|\bh^{*}_0[n]\bff_0[n-\rmd]|^2}{\rmQ_{\rmD}\displaystyle\sum_{F_i\in\Pi_f}{|\rmD_i|^{-\alpha_{f}}|\bg^{*}_i[n]\bw_i[n]|^2}} \geq \Upsilon \right]
= \bbP\left[\left|\bh^{*}_0[n]\bff_0[n-\rmd]\right|^2 \geq \rmQ_{\rmD}I_{f}\Upsilon\right],
\end{eqnarray}
\normalsize
where $\Upsilon$ denotes the $\mathsf{SIR}$ threshold.\\
We use a Gauss Markov autoregressive model, shown in the literature to be reasonably accurate for small delays on the communications links \cite{FSMC1, turin2004},  to account for the temporal correlation in the system. The effective channel at the receiver, assuming a Gauss Markov model \cite{KobayashiDelay}, is expressed as a weighted sum of two effective channels, 
\begin{equation}\label{eqn:GaussMarkov}
\bh_0^*[n]\bff_0[n-\rmd] = \left(\eta\bh_0^*[n - \rmd] + \sqrt{1 - \eta^2}\bee^*[n] \right) \bff_0[n-\rmd], 
\end{equation}
where $\bee[n]$ is a Gaussian vector with $\cC\cN(0,1)$ entries, independent of $\bh_0$. The correlation coefficient $\eta$ is determined using Clarkes isotropic scattering model as $\eta =  \cJ_0(2\pi \rmd f_dT_s)$, where $f_d$ is the maximum Doppler spread and $T_s$ is the symbol duration. The Doppler spread is  $f_d = vf_c/c$, where $v$ is the relative velocity of the transmitter-receiver pair, $f_c$ is the carrier frequency, and $c$ is the speed of light.
\begin{Lemm}\label{lem:approx_1}
The desired instantaneous channel power is written as a function of the first term of the Gauss Markov model $\bh_0[n-\rmd]$, for large values of $\eta$ corresponding to low mobility 
\begin{eqnarray}\label{eqn:approx_1}
\left|\bh_0^*[n]\bff_0[n-\rmd]\right|^2 &\approx& \eta^2 \left|\bh_0^*[n-\rmd]\bff_0[n-\rmd]\right|^2.
\end{eqnarray}
\end{Lemm}
\begin{Lemm}\label{lem:quant_apprx}
The desired effective channel power  $\left|\bh_0^*[n]\bff_0[n-\rmd]\right|^2 $ can be expressed as the weighted sum of two independent random variables, $Z = \eta^2\left(\left(1-\delta\right)X + Y\right)$, where $X$ and $Y$ are distributed as $\chi_{2(N_b - 1)}(1)$ and  $\chi_{2}(1)$, respectively. The CDF of the effective channel power is given by
\small
\begin{eqnarray}
 F_{Z}(z) &=& 1 - c_2\exp{\left(-\frac{z}{\kappa_2}\right)} + c_1\exp{\left(-\frac{z}{\kappa_1}\right)} \displaystyle\sum_{i = 0}^{N_b -2}{\sum_{\ell =0 }^{i}{\frac{\delta^i}{(i - \ell)!}\left(\frac{z}{\kappa_1}\right)^{i - \ell}}}, 
%
%
\end{eqnarray}
\normalsize
where $c_1 = \left(1 - \delta\right)\left(\frac{1}{\delta}\right)^{N_b-1}$, $\kappa_1 = 2\eta^2\left(1 - \delta \right)$, $c_2 =\left( \frac{1}{\delta}\right)^{N_b - 1}$, and $\kappa_2 = 2\eta^2$. 
\end{Lemm}
\begin{proof}
Using Gersho's conjecture on asymptotic quantization \cite{Mukkavilli2003, Zhou2005},  it was shown in \cite{Yoo2007} that $|\bh_0[n]|^2{\sin^2{\left(\angle{\left(\bh_0[n], \bv\right)}\right)}}$ has a Chi squared distribution $\chi_{2(N_b-1)}(\delta)$, and that $|\bh_0[n]|^2{\cos^2{\left(\angle{\left(\bh_0[n], \bv\right)}\right)}}$ can be expressed as the sum of two independent Chi squared random variables, $\chi_{2(N_b-1)}(1-\delta)$ and $\chi_{2}(1)$, respectively. Gersho's conjecture assumes that each quantization cell is a  Voronoi region of a spherical cap with a surface area equal $2^{-\rmB}$ of the total surface area of the unit sphere. This approximation has been shown to closely model the performance of codebook design techniques such as random vector quantization \cite{Yoo2007}, and is  used in the literature to analyze the performance of limited feedback systems\cite{Yoo2007, Jindal2006, Zhang2009}. The distribution of the sum of two independent Chi squared random variables with $2(N_b - 1)$ and $2$ degrees of freedom, respectively, follows from \cite{Simon2006}. 
\end{proof}

Figure \ref{fig:verifyApprox} compares the empirical CDF of the effective received power of $\left|\bh_0^*[n]\bff_0[n-\rmd]\right|^2$, with that of $\eta^2|\bh_0^*[n-\rmd]\bff_0[n - \rmd]|^2$ obtained in Lemma 2,  for increasing velocities $v \in \{10, 20, 30, 40, 50\}$ km/h, for a $4\times 1$ multiple input single output (MISO) downlink system, with a codebook size of B = 6. We observe that the approximation tightly follows the empirical CDF, with tolerable error up to moderate mobile speed of 50 km/h. 

\begin{theorem}\label{th:Psuccess}
\normalsize The probability of successful reception at a mobile cellular user, in the presence of cross tier interference, for a limited feedback MISO system, over temporally correlated channels, is given by
\small
\begin{equation}\label{eqn:theorem_1}
\bbP\left[\mathrm{SIR} \geq \Upsilon \right]  = A_1\omega_1{(\Upsilon)}e^{-\omega_1(\Upsilon)} + A_2e^{-\omega_1(\Upsilon)} + c_2e^{-\omega_2(\Upsilon)},
\end{equation}
\normalsize
where \small$$A_1 = c_1\sum_{i = 1}^{N_b - 2}{\displaystyle\sum_{\ell = 0}^{i - 1}{\frac{\delta^i}{(i-\ell)!}(-1)^{i - \ell} \prod_{m = 0}^{i - \ell - 1}{\left(\frac{2}{\alpha_f} - m \right)}}},\quad \omega_1(\Upsilon) = \lambda_fC_f\left(\frac{\Upsilon\rmQ_\rmD}{\kappa_1}\right)^{2/\alpha_f},$$ 
\begin{equation}\label{eqn:constants}
A_2 = -c_1\displaystyle\sum_{i = 0}^{N_b - 2}{\delta^i},\quad \omega_2(\Upsilon) = \lambda_fC_f\left(\frac{\Upsilon\rmQ_\rmD}{\kappa_2} \right)^{2/\alpha_f}, \mbox{ and    } C_f = \frac{2\pi}{\alpha_f}\Gamma\left(\frac{2}{\alpha_f}\right)\Gamma\left(1 - \frac{2}{\alpha_f}\right).\end{equation}\normalsize
\end{theorem}
\begin{proof}
The proof is provided in Appendix \ref{app:app_1}.
\end{proof}

Using Theorem 1,  we obtain an estimate of the maximum contention density $\lambda_f$ of femtocells in the network, given a probability of outage constraint $\epsilon$. 
\begin{Coroll}
The maximum femtocell contention density $\lambda_f(\rmD)$ for which the probability of successful reception satisfies the maximum outage probability constraint $\bbP\left[\mathsf{SIR} \geq \Upsilon \right] \geq 1 - \epsilon$, given $\Upsilon$, is the solution to 
\begin{equation}\label{eqn:density}
A_1\omega_1{(\lambda_f)}e^{-\omega_1(\lambda_f)} + A_2e^{-\omega_1(\lambda_f)} + c_2e^{-\omega_2(\lambda_f)} \geq 1 - \epsilon.
\end{equation}
In the high resolution regime, for small values of $\delta$ corresponding to a large codebook size $2^\rmB$, $\lambda$ is  given by
\small
\begin{equation} \label{eqn:lambda_f}
\lambda_f \leq \frac{\mathrm{LambertW}\left(-\frac{1-\epsilon}{A_1\exp(\frac{A_2 + c_2}{A_1})}\right) + \frac{A_2}{A_1} + \frac{c_2}{A_1}}{-C_f\left(\frac{\Upsilon\rmQ_\rmD}{\kappa_1}\right)^{\delta_f}},
\end{equation}
\normalsize
where LambertW denotes the Lambert W function that solves the equation  $W\exp(W) = X$ for $W$ as a function of $X$.
\end{Coroll}
\begin{proof}
The proof is provided in Appendix \ref{app:app_2}.
\end{proof}
Estimating the maximum number of femtocells as a function of the distance of the mobile user from the macro base station allows the network designer to put a limit on the number of femtocells to be deployed in each area. The number of femtocells increases exponentially with the number of antennas at the macro base station and the quantization size for the limited feedback beamforming system.
\section{Rate Backoff}\label{sec:Backoff}
When the estimated $\mathsf{SIR}^\sft$ at $B_0$ exceeds the estimated $\mathsf{SIR}$ at $M_0$, due to CSI mismatch, the transmitted rate exceeds the supported rate of the channel. This results in a packet outage. To decrease the probability of outage, and increase the probability of successful transmission, we implement rate backoff at the transmitter. Instead of transmitting at a rate corresponding to the estimated $\mathsf{SIR}^\sft$, the base station computes a backoff signal-to-interference ratio, $\mathsf{SIR}^\sfb\left(\mathsf{SIR}^\sft\right)$, which backs off the estimated $\mathsf{SIR}^\sft$ by a multiplicative factor $\beta$ 
\begin{equation}
\mathsf{SIR}^\sfb = \beta\;\mathsf{SIR}^\sft.
\end{equation}
\normalsize
The backoff $\mathsf{SIR}^\sft$ is chosen such that the average achievable goodput is maximized. Under the backoff model, the average goodput is written as
\small
\begin{equation}
\bar{\Lambda}^{\sfb}= \bbE\Big[\log_2\left(1 + \mathsf{SIR}^{\sfb}\right)\mathbb{P}\big[\mathsf{SIR}^{\sfb} \leq \mathsf{SIR}\big]\Big].
\end{equation}
\normalsize
Setting the backoff factor, $\beta$, to a value close to 1 provides a goodput rate close to that achieved without rate adaptation. A smaller $\beta$ value, however, might be conservative and lead to a low average achievable rate. We are interested in computing the optimal value for $\beta \in \left[0\;1\right]$ such that $\bar{\Lambda}^{\sfb}$ is maximized. 
The average goodput $\bar{\Lambda}^b$ is a continuous and differentiable function in $\beta$. Thus there exists an optimal backoff factor $\beta^*$ such that 
\begin{equation}
\beta^* = \argmax_{\beta} {\bar{\Lambda}^b}.
\end{equation}
\normalsize
We first compute the optimal backoff factor due to feedback delay, without taking into consideration the femtocell interference. We then extend the analysis to include both delay and interference.


\subsection{Rate backoff with feedback delay}
In the absence of interference, the feedback delay is the only source of CSI mismatch at the transmitter. The rate supported by the channel is 
\small
\begin{eqnarray}
\rmR[n] &=& \log_2\left(1 + \mathsf{SIR} \right)=  \log_2\left(1 + \frac{|\bh^*_0[n]\bff_0[n-d]|^2}{\rmQ_\rmD } \right),
\end{eqnarray}
\normalsize
where, with no cross tier interference, the pathloss ratio is $\rmQ_\rmD = \frac{1}{\rho_m\rmD^{-\alpha_m}}$.\\
The estimated $\mathsf{SIR}^\sft$ at the transmitter is 
 $\mathsf{SIR}^\sft =\frac{|\bh_0^*[n-\rmd]\bff_0[n-\rmd]|^2}{\rmQ_\rmD}$.

In the absence of interference, the probability of successful reception at the receiver is given by the CCDF of the effective channel power
\small
\begin{eqnarray}
\nonumber  \bbP\left[\mathsf{SIR} \geq \Upsilon \right]  = 
 c_2\exp{\left(-\frac{z}{\kappa_2}\right)} - c_1\exp{\left(-\frac{z}{\kappa_1}\right)} \sum_{i = 0}^{N_b -2}{\sum_{\ell =0 }^{i}{\frac{\delta^i}{(i - \ell)!}\left(\frac{z}{\kappa_1}\right)^{i - \ell}}}.
\end{eqnarray}
\normalsize
To maximize the average goodput  with the suboptimal rate backoff function $\beta\Upsilon$, one needs to find a factor $\beta^*$ such that the integral on the right hand side of 
\small
\begin{eqnarray}\label{eqn:avggoodputbackoff}
\bar{\Lambda}^\rmb =  \int_{0}^{\infty}{\rmR^\rmt(\beta\Upsilon) \bbP\left(\mathsf{SIR} \geq \beta\Upsilon\right) f_\Upsilon(\Upsilon)d\Upsilon}
\end{eqnarray}
\normalsize
is maximized. In (\ref{eqn:avggoodputbackoff}), maximizing the argument of the integration results in maximizing the integration. We interchange the integration with the maximization, and the problem becomes finding $\beta^*$ such that
\begin{eqnarray}
\beta^* &=& \argmax_{\beta}\left[ \log_2\left(1 + \beta \Upsilon\right) \bbP\left(\mathsf{SIR} \geq \beta\Upsilon \right)\right].
\end{eqnarray}
\begin{theorem}
The backoff factor $\beta^*$ that maximizes the average goodput for a limited feedback beamforming MISO system, with delay on the feedback channel, is the non-trivial solution of 
\small
\begin{align}\label{eqn:theorem_2}
\nonumber&\log_2\left(1 + \frac{\beta\Upsilon}{\rmQ_\rmD}\right) \left[c_1e^{(k_1\beta)}\left(k_1(f_1(\beta) + \frac{1 - \delta^{(N_b - 1)}}{1 -\delta}) - f_2(\beta) \right)  - c_2k_2e^{(k_2\beta)}\right] =\\
&\left(\frac{1}{1 + \frac{\beta\Upsilon}{\rmQ_\rmD}}\right)\left[c_1e^{\left(k_1\beta\right)}\left(k_2\left(f_1\left(\beta\right) + k_3\frac{1 - \delta^{(N_b - 1)}}{1 -\delta}\right) \right) +c_2k_3e^{(k_2\beta)}\right],
\end{align}
where $f_1(\beta) = \displaystyle\sum_{i = 1}^{N_b - 2}{\sum_{\ell =0 }^{i-1}{\frac{\delta^i}{(i - \ell)!}\left(\frac{\beta\Upsilon}{\kappa_1}\right)^{i-\ell}}}$, $f_2(\beta) = \displaystyle\sum_{i = 1}^{N_b - 2}{\sum_{\ell =0 }^{i-1}{\frac{\delta^i}{(i - \ell)!}\left(\frac{\beta\Upsilon}{\kappa_1}\right)^{i-\ell}\left(\frac{i-\ell}{\beta}\right)}}$,\\
  $k_1 =  - \frac{\Upsilon}{\kappa_1}$,  $k_2 =  - \frac{\Upsilon}{\kappa_2}$, and $k_3 = \frac{\Upsilon}{\rmQ_\rmD\log(2)}$.
\normalsize

For large codebook sizes, setting $\kappa_2 = \kappa_1$,  the backoff factor $\beta$ can be computed as the largest root of the polynomial in $\beta$ of order $N_b-1$
\begin{equation}\label{eqn:theorem_2_approx}
C_1 P_1^{N_b-1}\beta)  - C_2P_2^{N_b-2}(\beta) + C_3P_3^{N_b - 1}(\beta) + C_4\beta + C_5 = 0,
\end{equation}
with \small$P^{N_b-1}_1(\beta) =  \displaystyle\sum_{i = 0}^{N_b - 2}{\sum_{\ell =0 }^{i}{\frac{\delta^i}{(i - \ell)!}\beta\left(\frac{\beta\Upsilon}{\kappa_1}\right)^{i-\ell}}},  P^{N_b-1}_2(\beta) = \frac{P^{N_b-1}_1(\beta)}{\beta},$
$P^{N_b-1}_3(\beta) =\displaystyle \sum_{\ell = 0}^{N_b-2}{\frac{\beta}{\ell!}{\left(\frac{\beta\Upsilon}{\kappa_1}\right)^\ell}}$\normalsize, and \\\small $C_1 =(1 - \delta)(-\frac{\Upsilon}{\kappa_1}-1)\frac{\Upsilon}{\rmQ_\rmD}$, $C_2 = -(1-\delta)$, $C_3 = \Upsilon(1-\delta)\delta^{N_b-2}$,  $C_4 =\left (1 + 2\frac{\Upsilon}{\kappa_1}\right)\frac{\Upsilon}{\rmQ_\rmD}$, and $C_5 = 2\frac{\Upsilon}{\rmQ_\rmD\log(2)}+\frac{\Upsilon}{\kappa_1}$.\normalsize   
\end{theorem}
\begin{proof}
The proof is provided in Appendix \ref{app:app_3}.
\end{proof}
The backoff factor $\beta^*$ in the presence of feedback delay is a function of the distance $\rmD$ and the correlation coefficient $\eta$. It also varies with the $\mathsf{SIR}$ threshold at the receiver and the number of feedback bits B,  relative to the number of antennas at the base station, $N_b$. The backoff factor $\beta$ admits an algebraic solution for $N_b \leq 6$ by the Abel-Ruffini theorem.

\subsection{Rate backoff with cross tier interference and feedback delay}
In the presence of uncoordinated interference at the mobile user $M_0$, the probability of successful reception is given by (\ref{eqn:theorem_1}). The backoff factor  $\beta^*$ computation follows similarly to the case of rate backoff with feedback delay. 

\begin{theorem}
The backoff factor $\beta^*$, that maximizes the average goodput for a limited feedback beamforming MISO system, in the presence of cross tier interference with maximum density $\lambda_f(\rmD)$ is the non-trivial solution of
\small
\begin{align}
&\nonumber \left(\frac{1}{1 + \frac{\beta\Upsilon}{\rmQ_\rmD}}\right)\left[ \left(A_1\omega_1(\beta\Upsilon) + A_2\right)e^{-\omega_1(\beta\Upsilon)}  + c_2e^{-\omega_2(\beta\Upsilon)}\right] = \\
&-\log_2\left(1 + \frac{\beta\Upsilon}{\rmQ_\rmD}\right) \left[\frac{\delta_f}{\beta}e^{-\omega_1(\beta\Upsilon)}\left((A_1 - A_2)\omega_1(\beta\Upsilon) - A_1\omega_1^2(\beta\Upsilon)\right) - \frac{c_2\delta_f}{\beta}\omega_2(\beta\Upsilon)e^{-\omega_2(\beta\Upsilon)}\right].
\end{align}
\normalsize
For large codebook sizes,  the backoff factor $\beta^*$ is computed as the $\delta_f$th square root of the roots of the polynomial in $\beta^{\delta_f}$ of order $2$, 
\small
\begin{equation}
A_1\delta_f\omega_1^2(\beta\Upsilon) + \delta_f\left(c_2 - A_1(1 + \log(2)) + A_2  \right)\omega_1(\beta\Upsilon) - \frac{A_2 + c_2}{\log(2)},
\end{equation}
\normalsize
such that $\beta^* \in \left[0\;1 \right]$.
\end{theorem}
\begin{proof}
The proof follows similarly to the proof of Theorem 2, with the expression of the probability of successful reception given by $(\ref{eqn:theorem_1})$. 
\end{proof}
The backoff factor $\beta^*$  is a decreasing function of the velocity of $M_0$ and thus an increasing function of $\eta$. 
This implies that as the velocity of the mobile user increases, the average achievable rate in the system, with rate backoff, decreases, and comes closer to the average achievable rate without backoff. It is  further a decreasing function of the femtocell contention density $\lambda_f$, and the distance $\rmD$ between $M_0$ and $B_0$. Thus, as the $\mathsf{SIR}$ at $M_0$ increases, the average achievable rate with backoff increases. 

The backoff factor $\beta^*$ is an increasing function of $\delta$. Increasing the number of antennas, and fixing the feedback channel rate $\rmB$, the average achievable rate with backoff increases relative to that without backoff. Similarly, increasing the rate on the feedback channel, for a fixed number of transmit antennas, results in a smaller backoff factor and a lower backoff rate.

\section{Simulation Results and Discussion}\label{sec:Simul}
We consider a single macrocell of radius $R_c = 1$km, overlayed with a tier of randomly distributed femtocells. The femtocells are distributed according to a homogeneous PPP with an average density of $\lambda_f$ femtocells per cell-site.  We adopt a distance based pathloss model corresponding to the IMT-2000 channel model \cite{IMT2000} for outdoor and indoor pathloss. For analytical simplicity, we do not consider random Log-normal shadow fading. We consider frequency flat Rayleigh fading with individual complex entries distributed as $\mathcal{CN}(0,1)$. We assume fixed wall partition losses corresponding to indoor-to-outdoor and outdoor-to-indoor propagation, equal to $5$ dB. The outdoor and indoor-to-outdoor pathloss exponents are set to $3.8$ and the carrier frequency is 2 GHz.  The users are uniformly distributed inside each cell.

We start by plotting the empirical CDF of the probability of outage at a mobile user with cross tier interference and delay. For a user average velocity of $20$km/h and a target $\mathsf{SIR}$ threshold of 5 dB, Figure \ref{fig:verifyApproxPout} compares the empirical CDF of the probability of outage versus $\mathsf{SNR}$ at the mobile user, to the expression of the probability of outage obtained using Theorem 1. The figure shows that the bound in Theorem 1 closely represents the probability of outage for all possible locations of the mobile user inside the macrocell. 

To get a practical sense of the number of femtocells obtained from Equation (\ref{eqn:density}), we plot, in Figure \ref{fig:lambda_f}, the maximum number of transmitting femtocells $N_f = \pi R_c^2 \lambda_f$ as a function of  $\mathsf{SNR}$ at $M_0$, with a 10 percentile outage probability requirement, for increasing $\delta = 2^{-\frac{\rmB}{N_b-1}}$. The density $\lambda_f$ is averaged over $1000$ uniformly distributed users in the macrocell, for an average user velocity of $20$Km/h. As the number of  feedback bits $\rmB$ increases, for example, for $N_b = N_f = 4$ antennas, the number of femtocell interferers that can be allowed in the system without violating the probability of outage requirement increases, this increase is more discernable for higher feedback rate, $\rmB > 6$. Similarly, as the number of antennas at the base station $N_b$ increases, for a fixed number of feedback bits, $\lambda_f$ increases.

Figures \ref{fig:ratebackoffDelay} and \ref{fig:ratebackoff} illustrate the rate backoff concept applied at the transmitter to maximize the average goodput, respectively for delay and cross tier interference and delay. Figure \ref{fig:ratebackoffDelay} plots the ergodic goodput for a mobile user velocity of $20$Km/h, $N_b = N_f = 2$ and delay of 2 time frames.  The goodput achieved when the probability of outage is equal to zero, i.e. the rate supported by the channel (the throughput) as a function of the $\mathsf{SNR}$ at the receiver is shown for comparison. We observe that the gap between the goodput and throughput is considerable. This gap increases with $\mathsf{SNR}$ as the rate of increase for the goodput depends on the probability of outage. Applying rate backoff at the transmitter recovers the rate of increase of the achievable rate with the $\mathsf{SNR}$ to that of the throughput. The gap between the backoff rate and that of the throughput remains because of the conservative transmitted rate, when the estimated $\mathsf{SIR}^\sft$ falls short of the estimated $\mathsf{SIR}$ at $M_0$. 
We also plot the backoff goodput using the approximation in (22) for small $\delta$ values. The approximation is shown to be accurate for $N_b = 2$ antennas. It gives, as expected, a lower bound on the backoff factor $\beta^*$ and the backoff goodput. 

Figure \ref{fig:ratebackoff} shows the ergodic rate for $No_f = 95$. It compares the average goodput achieved with limited feedback beamforming, with and without rate backoff, to that achieved using open loop random beamforming \cite{viswanath2002}, for $N_b = N_f = 4$ transmit antennas, and $\rmB = 5$ feedback bits. The throughput is shown for comparison.
We observe that limited feedback beamforming achieves an average rate gain of 5 dB over random beamforming for $\bbE[\rmR[n]] = \log_2\left(1 + \mathsf{SINR}\right)$. For the goodput $\bar{\Lambda}$, this gain is on average 4 dB. 
This suggests that limited feedback beamforming, although prone to errors due to delay and quantization, is still a viable option for macrocell transmission use in heterogeneous networks MISO systems. 
Similarly to the system without interference, applying rate backoff at the transmitter recovers the rate of increase of the achievable rate to that of the throughput. The achievable rate with backoff is very close to that achieved with random beamforming, in the absence of outage. 

Figure \ref{fig:alphadeltavel} plots the backoff factor $\beta^*$ as a function of the velocity of the mobile user and the number of feedback bits relative to the number of transmit antennas, $\delta$. The figure shows that the backoff factor $\beta$ is a decreasing function of the velocity $v \in [20\; 60]$ km/h. Rate backoff performs better for high temporal correlation coefficient $\eta$, corresponding to pedestrian or low speed mobile users. This is reasonable because the approximation in (\ref{eqn:approx_1}) assumes low mobility. $\beta^*$ is an increasing function of $\delta = 2^{-\frac{\rmB}{N_b-1}}$. Increasing the number of antennas for fixed feedback rate increases the backoff rate achieved. Similarly, decreasing the number of feedback bits, for a fixed number of antennas, decreases the backoff factor. 

\section{Conclusion}\label{sec:Conclusion}
In this paper, we proposed a MISO limited feedback system for use on the downlink of a heterogeneous cellular system. We evaluated the effect of uncoordinated cross tier interference and feedback delay on the performance of the system. 
We derived a closed form expression for the probability of outage as a function of the distance from the mobile receiver to its base station, the velocity of the user and the femtocell density. We showed that, although prone to rate outage due to the inaccurate CSI at the transmitter, limited feedback MISO is a good candidate for use on the downlink of heterogeneous cellular networks. To maximize the average achievable rate, we proposed rate backoff techniques at the transmitter. We showed that rate backoff maximizes the goodput achieved, and hence increases the long term achievable rate in the system. Future work includes applying rate backoff to the second tier network considering both cross tier and same tier interference. 

\appendices

\section{Proof of Theorem 1} \label{app:app_1}
The probability of successful transmission is computed as
\small
\begin{subequations}
\begin{eqnarray}
\nonumber \bbP\left[\mathsf{SIR} \geq \Upsilon \right] &=& \bbP\left[|\bh^{*}_0[n]\bff_0[n-d]|^2 \geq \Upsilon \rmQ_\rmD I_{f,c}\right]\\
&\stackrel{a}=& \nonumber \int_{0}^{\infty}{\left[c_2\exp{\left(-\frac{s \rmQ_\rmD \Upsilon}{\kappa_2}\right)} - c_1\exp{\left(-\frac{s \rmQ_\rmD \Upsilon}{\kappa_1}\right)} \sum_{i = 0}^{N_b -2}{\sum_{\ell =0 }^{i}{\frac{\delta^i}{(i - \ell)!}\left(\frac{s \rmQ_\rmD\Upsilon}{\kappa_1}\right)^{i - \ell}}}\right]\rmd\bbP\left(I_{f,c} \leq s\right)}\\
&=& \int_{0}^{\infty}{\left[c_2\exp{\left(-\frac{s \rmQ_\rmD \Upsilon}{\kappa_2}\right)} - c_1\exp{\left(-\frac{s \rmQ_\rmD \Upsilon}{\kappa_1}\right)} \sum_{i = 0}^{N_b -2}{{\delta^i}}\right]\rmd\bbP\left(I_{f,c} \leq s\right)}\\
&-& \int_{0}^{\infty}{\left[ c_1\exp{\left(-\frac{s \rmQ_\rmD \Upsilon}{\kappa_1}\right)} \sum_{i = 1}^{N_b - 2}{\sum_{\ell =0 }^{i-1}{\frac{\delta^i}{(i - \ell)!}\left(\frac{s \rmQ_\rmD \Upsilon}{\kappa_1}\right)^{i - \ell}}}\right]\rmd\bbP\left(I_{f,c} \leq s\right)}
\end{eqnarray}
\end{subequations}
\normalsize
where step (a) follows from conditioning on $I_{f,c}$ and substituting for the CCDF of $|\bh^{*}_0[n]\bff_0[n-d]|^2$. 

The integral in (24a) is expressed in terms of the Laplace transform of  $I_{f,c}$, $\cL_{I_{f,c}}(\theta) = \bbE\left[e^{-I_{f_c}\theta}\right] =\exp{\left(-\lambda_fC_f\theta^{\delta_f}\right)}$ (from \cite{Hunter2008}), evaluated at $\theta = \frac{\rmQ_\rmD\Upsilon}{\kappa_2}$ and $\theta =  \frac{\rmQ_\rmD\Upsilon}{\kappa_1}$, respectively, for $C_f = \pi\delta_f\Gamma(\delta_f)\Gamma(1 - \delta_f)$,
\small
\begin{eqnarray}
\nonumber \int_{0}^{\infty}{\left[c_2\exp{\left(-\frac{s \rmQ_\rmD \Upsilon}{\kappa_2}\right)} - c_1\exp{\left(-\frac{s \rmQ_\rmD \Upsilon}{\kappa_1}\right)} \sum_{i = 0}^{N_b -2}{{\delta^i}}\right]\rmd\bbP\left(I_{f,c} \leq s\right)}
=\nonumber  c_2 \cL_{I_{f,c}}\left(\frac{\Upsilon \rmQ_\rmD}{\kappa_2} \right) - \left(c_1\sum_{i= 0}^{N_b -2}{\delta^i}\right)\cL_{I_{f,c}}\left(\frac{\Upsilon \rmQ_\rmD}{ \kappa_1} \right) \\
\nonumber=  c_2 \exp\left(-\lambda_f C_f\left(\frac{\Upsilon \rmQ_\rmD}{\kappa_2}\right)^{\delta_f}\right)  - \left(c_1\sum_{i= 0}^{N_b -2}{\delta^i}\right)\exp\left(-\lambda_f C_f\left(\frac{\Upsilon \rmQ_\rmD}{\kappa_1}\right)^{\delta_f}\right).
\end{eqnarray}
\normalsize
The integral in  (24b) is evaluated as
\small
\begin{eqnarray}
\nonumber \int_{0}^{\infty}{\left[ c_1\exp{\left(-\frac{s \rmQ_\rmD \Upsilon}{\kappa_1}\right)} \sum_{i = 1}^{N_b - 2}{\sum_{\ell =0 }^{i-1}{\frac{\delta^i}{(i - \ell)!}\left(\frac{s \rmQ_\rmD \Upsilon}{\kappa_1}\right)^{i - \ell}}}\right]\rmd\bbP\left(I_{f,c} \leq s\right)}
\nonumber \stackrel{a}=c_1 \sum_{i = 1}^{N_b - 2}{\sum_{\ell =0 }^{i-1}{\frac{\delta^i}{(i - \ell)!}\left(-\frac{\Upsilon \rmQ_\rmD}{\kappa_1}\right)^{i-\ell}\frac{d^{i-\ell}}{d\theta^{i-\ell}}\cL_{I_{f,c}}\left(\theta\right)}}\\
\nonumber \stackrel{b}=c_1 \sum_{i = 1}^{N_b - 2}{\sum_{\ell =0 }^{i-1}{\frac{\delta^i}{(i - \ell)!}\left(-\frac{\Upsilon \rmQ_\rmD}{\kappa_1}\right)^{i-\ell}\left[-\lambda_f C_f \Pi_{m = 0}^{i-\ell - 1}{(\delta_f -  m)}\theta^{\delta_f - i + \ell}\exp\left(-\lambda_f C_f \theta^{\delta_f}\right)\right]}},
\end{eqnarray}
\normalsize
where (a) follows from the identity $\cL\left[x^kf(x)\right] = (-1)^kF^{(k)}(s)$ for the Laplace transform, with  $F^{(k)}(s)$ representing the $k$th derivative of $F(s)$. Step (b) follows from evaluating the $k$th derivative of the Laplace Transform  $\cL_{I_{f,c}}(\theta)$, \cite{Hunter2008}
\small
$$ \frac{\rmd^k}{\rmd \theta^k}  \exp{\left(-\lambda_fC_f\theta^{\delta_f}\right)} = - \left[\lambda_f C_f \Pi_{m = 0}^{k - 1}{(\delta_f -  m)}\theta^{\delta_f -k}\exp\left(-\lambda_f C_f \theta^{\delta_f}\right) \right] + \Theta(\lambda_f^2C_f^2\theta^{2\delta_f}),$$
\normalsize
and ignoring the second order terms. 

Combining the two parts of the integration, the probability of successful reception at $M_0$ is 
\small
\begin{eqnarray}
\nonumber \bbP\left[\mathrm{SIR} \geq \Upsilon \right]  &=&  c_2 \exp\left(-\lambda_f C_f\left(\frac{\Upsilon \rmQ_\rmD}{\kappa_2}\right)^{\delta_f}\right)  - \left(c_1\sum_{i= 0}^{N_b -2}{\delta^i}\right)\exp\left(-\lambda_f C_f\left(\frac{\Upsilon \rmQ_\rmD}{\kappa_1}\right)^{\delta_f}\right)\\ 
\nonumber &-& c_1 \sum_{i = 1}^{N_b - 2}{\sum_{\ell =0 }^{i-1}{\frac{\delta^i}{(i - \ell)!}\left(-\frac{\Upsilon \rmQ_\rmD}{\kappa_1}\right)^{i-\ell}\left[-\lambda_f C_f \Pi_{m = 0}^{i-\ell - 1}{(\delta_f -  m)}\theta^{\delta_f - i + \ell}\exp\left(-\lambda_f C_f \theta^{\delta_f}\right)\right]}}\\
&=& A_1\omega_1{(\Upsilon)}e^{-\omega_1(\Upsilon)} + A_2e^{-\omega_1(\Upsilon)} + c_2e^{-\omega_2(\Upsilon)},
\end{eqnarray}\normalsize
where $A_1$,  $A_2$, $\omega_1(\Upsilon)$, $\omega_2(\Upsilon)$, and $C_f$ are given in (\ref{eqn:constants}).
%
%
\section{Proof of Corollary 1}\label{app:app_2}
For small values of $\delta = 2^{-\frac{B}{N_b-1}}$, corresponding to large codebook sizes, relative to the number of transmit antennas, we make the approximation $\kappa_1 = \kappa_2 = 2\eta^2(1 - \delta)$.
This implies that $\omega_1(\lambda_f) = \omega_2(\lambda_f) = \lambda_fC_f\left(\frac{\Upsilon\rmQ_\rmD}{\kappa_1}\right)^{\delta_f}$. The approximation gives a lower bound on the probability of success and consequently an upper bound on the probability of outage, $c_2e^{-\omega_1(\lambda_f)} \leq c_2e^{-\omega_2(\lambda_f)}.$
The maximum femtocell contention density is derived from,
\small
\begin{align}
\nonumber &A_1\omega_1{(\lambda_f)}e^{-\omega_1(\lambda_f)} + A_2e^{-\omega_1(\lambda_f)} + c_2e^{-\omega_1(\lambda_f)} \geq 1 - \epsilon \longrightarrow\\
\nonumber& - A_1\left(-\omega_1{(\lambda_f)} - \frac{A_2}{A_1} - \frac{c_2}{A_1}\right)e^{-\omega_1(\lambda_f)} e^{- \frac{A_2}{A_1} - \frac{c_2}{A_1}} e^{\frac{A_2}{A_1} + \frac{c_2}{A_1}} \geq 1 - \epsilon
&\longrightarrow \quad C(\lambda_f)e^{C(\lambda_f)} \geq \frac{1 - \epsilon}{-A_1e^{\frac{A_2}{A_1} + \frac{c_2}{A_1}}},
\end{align}
\normalsize
where $C(\lambda_f) = -\omega_1{(\lambda_f)} - \frac{A_2}{A_1} - \frac{c_2}{A_1}$. It is given by the solution to the monotonically decreasing Lambert W function $\mathrm{LambertW}(\lambda_f) = C(\lambda_f)e^{C(\lambda_f)}$. Consequently,  $\lambda_f$ is given by (\ref{eqn:lambda_f}).

\section{Proof of Theorem 2}\label{app:app_3}
The expression for $\beta^*$ follows from taking the derivative with respect to $\beta$ of 
\begin{eqnarray}
\beta^* &=& \argmax_{\beta}\left[ \log_2\left(1 + \beta \Upsilon\right) \bbP\left(\mathsf{SINR} \geq \beta\Upsilon \right)\right]
\end{eqnarray}
and setting the derivative equal to zero. Under the definitions $f_1(\beta) = \sum_{i = 1}^{N_b - 2}{\sum_{\ell =0 }^{i-1}{\frac{\delta^i}{(i - \ell)!}\left(\frac{\beta\Upsilon}{\kappa_1}\right)^{i-\ell}}}$, $f_2(\beta) = \sum_{i = 1}^{N_b - 2}{\sum_{\ell =0 }^{i-1}{\frac{\delta^i}{(i - \ell)!}\left(\frac{\beta\Upsilon}{\kappa_1}\right)^{i-\ell}\left(\frac{i-\ell}{\beta}\right)}}$,\\
  $k_1 =  - \frac{\Upsilon}{\kappa_1}$,  $k_2 =  - \frac{\Upsilon}{\kappa_2}$, and $k_3 = \frac{\Upsilon}{\rmQ_\rmD\log(2)}$, the optimization expression is written as given in Theorem 2.
  
For large codebook sizes, we make the approximation $\kappa_2 = \kappa_1$. This approximation allows us to remove the exponential terms from the expression in (\ref{eqn:theorem_2}). Furthermore, we notice that
\small
$$f_2(\beta) =   \sum_{i = 0}^{N_b - 2}{\sum_{\ell =0 }^{i}{\frac{\delta^i}{(i - \ell)!}\left(\frac{\beta\Upsilon}{\kappa_1}\right)^{i-\ell}}} - \frac{1-\delta^{N_b-1}}{1-\delta} = P^{N_b-2}_2(\beta) - \frac{1-\delta^{N_b-1}}{1-\delta},$$
$$f_1(\beta) =  P^{N_b-2}_2(\beta)  -  \sum_{\ell = 0}^{N_b-2}{\frac{\beta}{\ell!}{(\frac{\beta\Upsilon}{\kappa_1})^\ell}} = P^{N_b-2}_2(\beta) - P^{N_b-2}_3(\beta).$$ 
\normalsize
For $\rmQ_\rmD$ large, we have $\log_2\left(1 + \frac{\beta\Upsilon}{\rmQ_\rmD}\right) \approx  \frac{\beta\Upsilon}{\rmQ_\rmD}$, and $\frac{1}{1 + \frac{\beta\Upsilon}{\rmQ_\rmD}} \approx 1$. The expression in (\ref{eqn:theorem_2_approx}) then follows by rearranging the terms in (\ref{eqn:theorem_2}). 
  

%
%
\bibliographystyle{IEEEtran}
\bibliography{IEEEabrv,ReferencesSalam}
\newpage

\centering{\textbf{\large{Figures}}}

\begin{figure}[h]
  \begin{center}
    \includegraphics[scale=0.65]{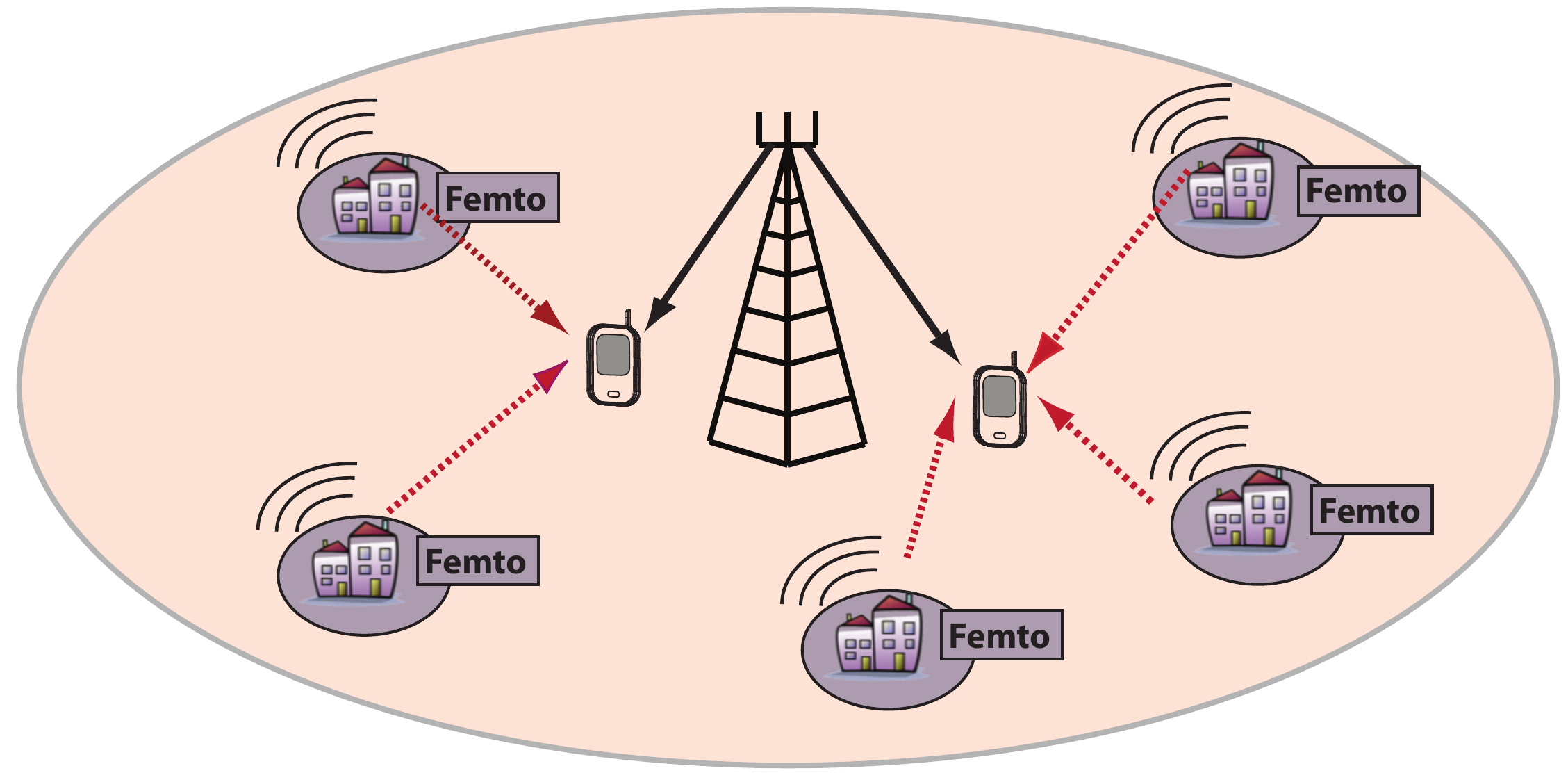}
    \caption{In the downlink scenario, mobile users $M_0$ and $M_1$ experience cross tier interference from small cells and femtocells randomly distributed inside the macrocell. The macro base station and the femtocells are equipped with multiple antennas. $M_0$ and $M_1$ have a single antenna.}
\label{fig:cellsysfemto}.
  \end{center}
\end{figure}

\begin{figure}[h]
  \begin{center}
    \includegraphics[scale=0.35]{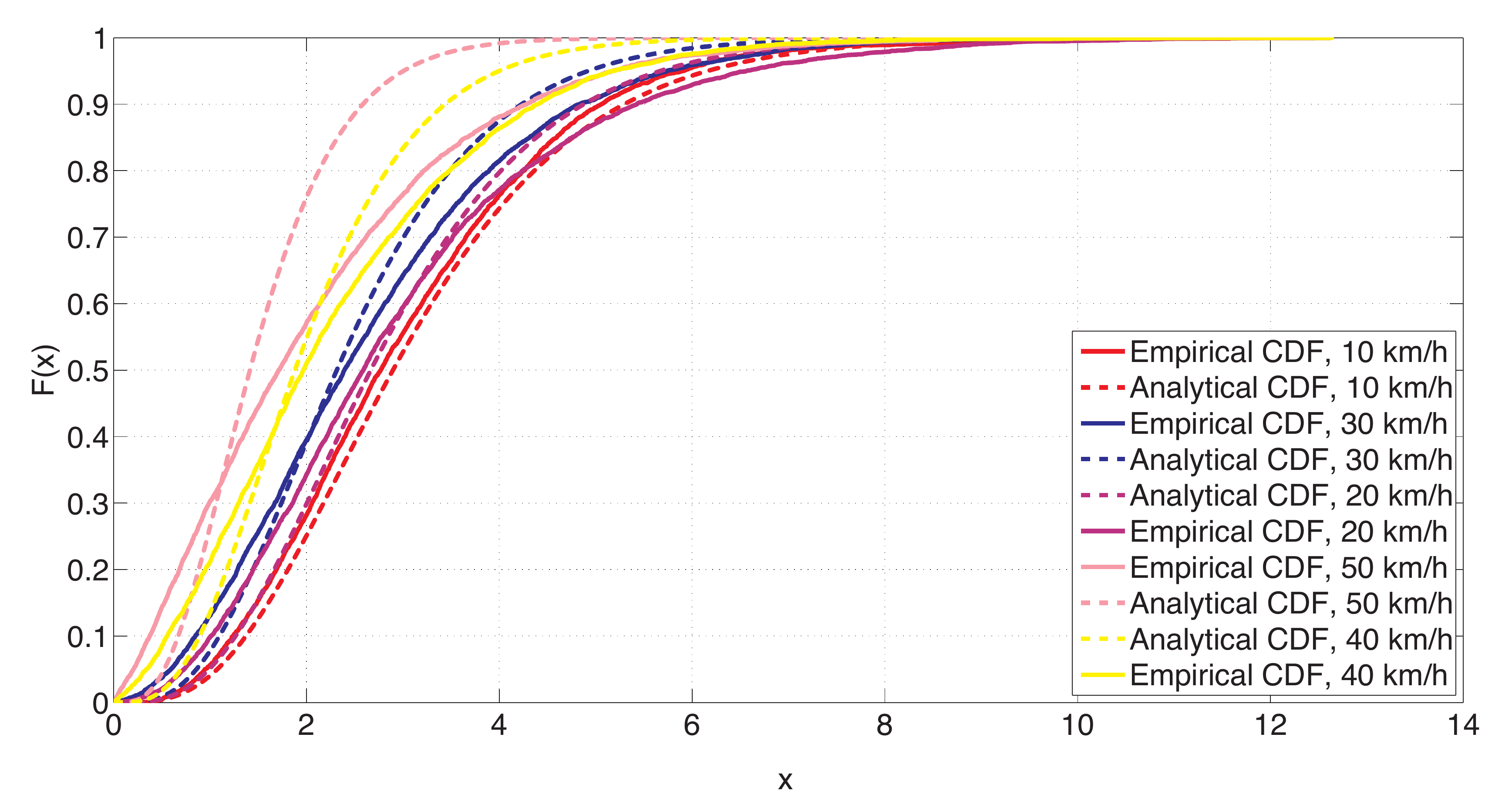}
    \caption{The CDF of the effective received power $|\bh^*_0[n]\bff_0[n-\rmd]|^2$ and of the approximation $\eta^2|\bh^*_0[n-\rmd]\bff_0[n-\rmd]|^2$, for increasing velocities of the mobile user $M_0$ $v \in \{10, 20, 30, 40, 50\}$ km/h, for a $4\times 1$ MISO downlink system with a codebook size of $\rmB = 6$, and delay $\rmd = 2$ time frames.}
\label{fig:verifyApprox}
  \end{center}
\end{figure}

%
%

\begin{figure}[h]
  \begin{center}
    \includegraphics[scale=0.35]{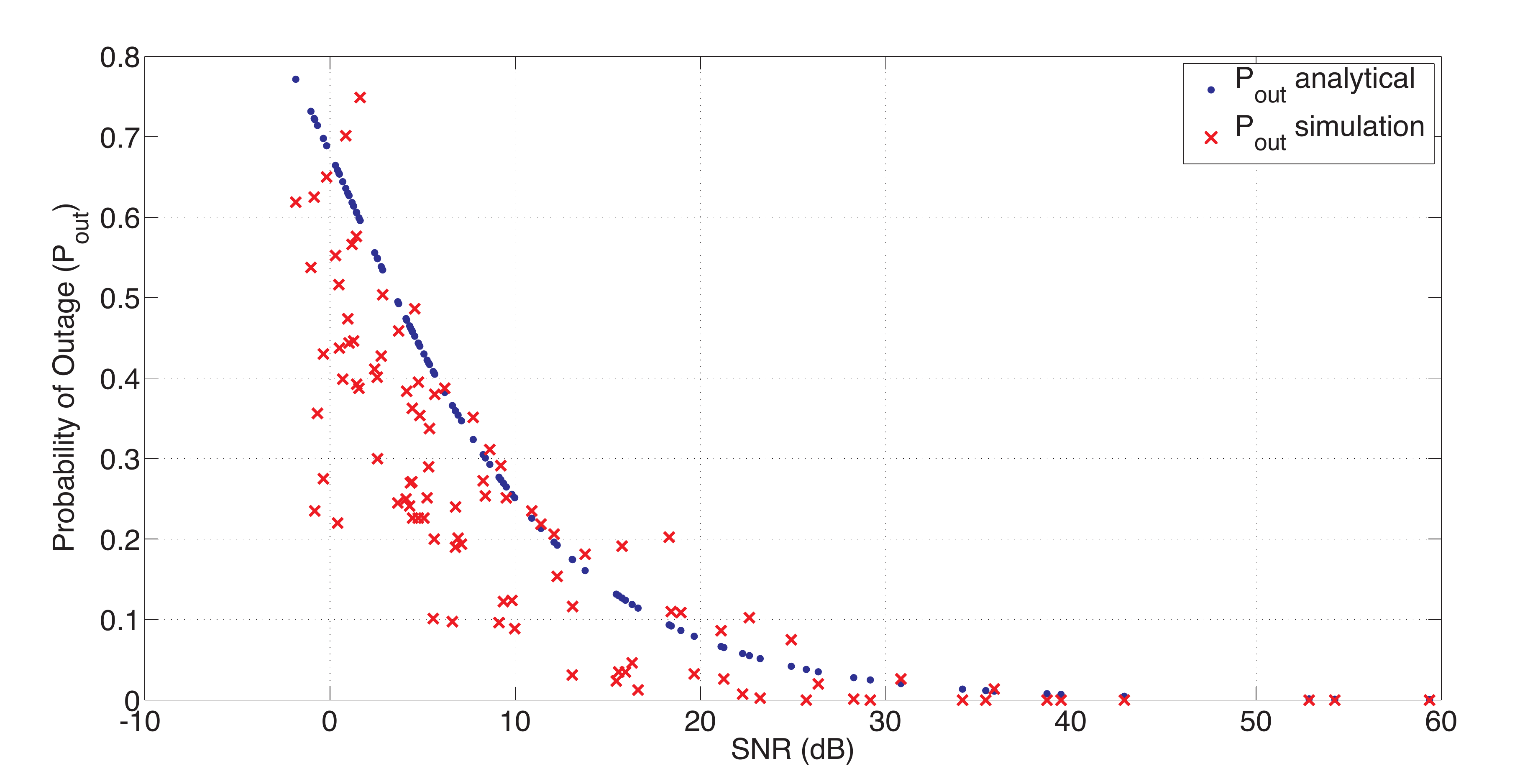}
    \caption{The probability of outage $P_{out} = \bbP[\mathsf{SIR}\leq \Upsilon]$ as a function of SNR or the distance $\rmD$ from $M_0$ to $B_0$ for $N_b = N_f = 4$,  codebook size of $\rmB = 5$, and mobile velocity of 20km/h. The analytical $P_{out}$ corresponds to that derived in Theorem 1. }
\label{fig:verifyApproxPout}
  \end{center}
\end{figure}

%
%

\begin{figure}[t]
  \begin{center}
    \includegraphics[scale = 0.50]{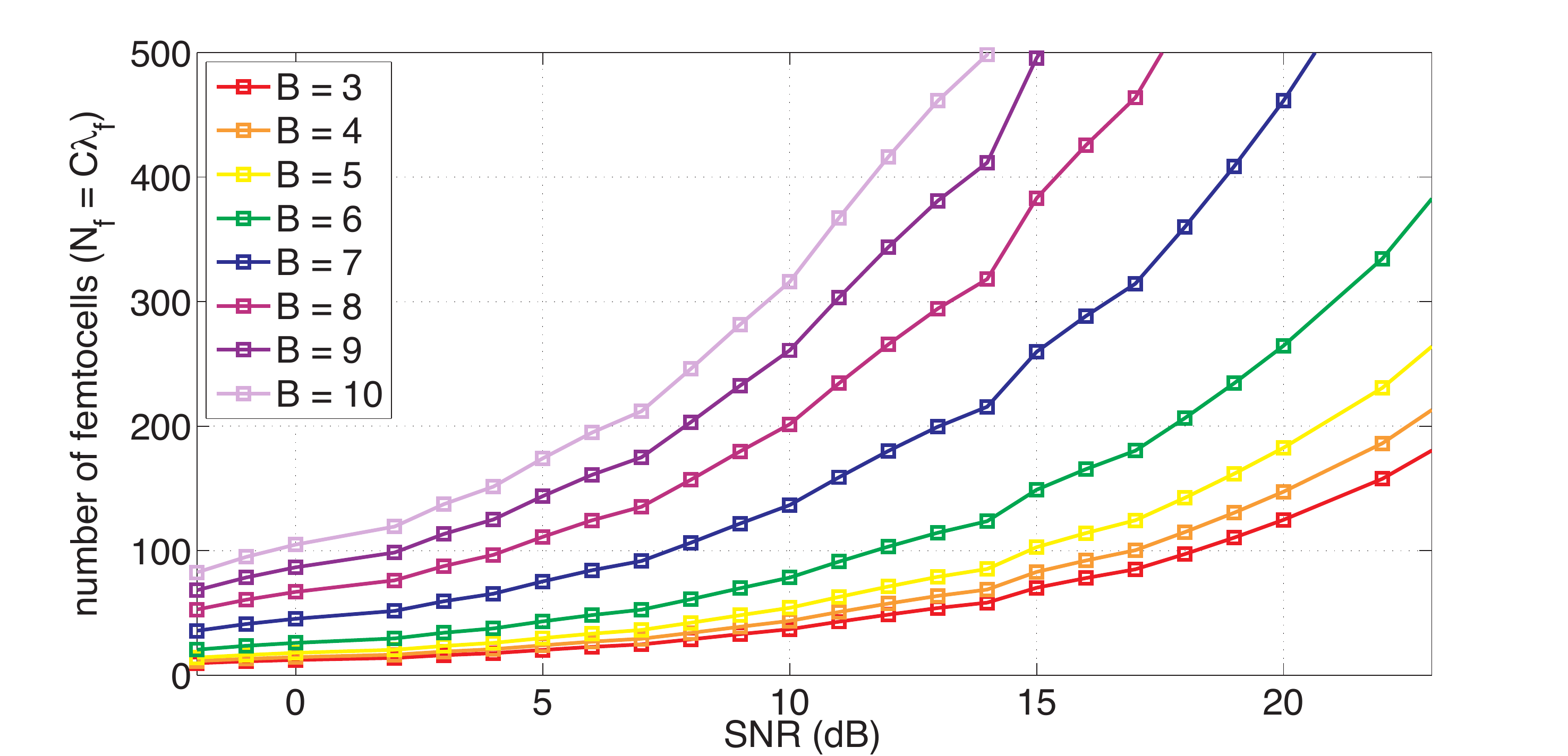}
    \caption{The maximum number of femtocells interferers in a macrocell of area of $C = \pi 1000^2$ for a $4\times1$ MISO limited feedback system with increasing number of feedback bits $\rmB$ for $N_b = N_f = 4$, velocity $v = 20$ km/h, and delay $\rmd = 2$ time frames.}
    \label{fig:lambda_f}
  \end{center}
\end{figure}

\begin{figure}[t]
  \begin{center}
  \vspace{-10pt}
    \includegraphics[scale = 0.5]{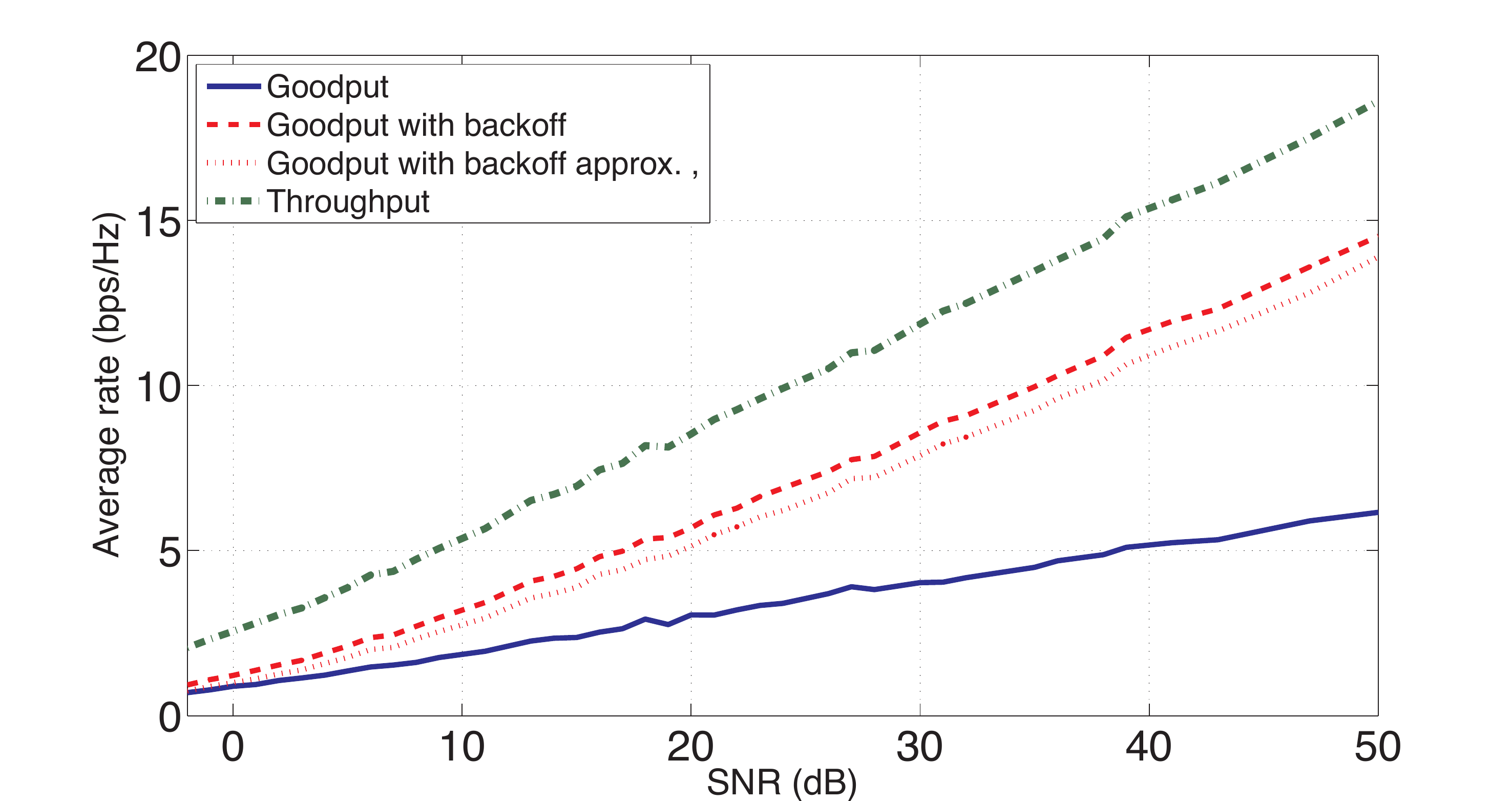}
    \caption{The average goodput for a MISO limited feedback system with $N_f = N_b = 2$, and $\rmB = 3$ as a function of SNR with and without rate backoff. The femtocell density is assumed $N_{of} = 95$ femtocells per cell site. }
    \vspace{-5pt}
    \label{fig:ratebackoffDelay}
  \end{center}
\end{figure}

\begin{figure}[t]
  \begin{center}
  \vspace{-10pt}
    \includegraphics[scale = 0.47]{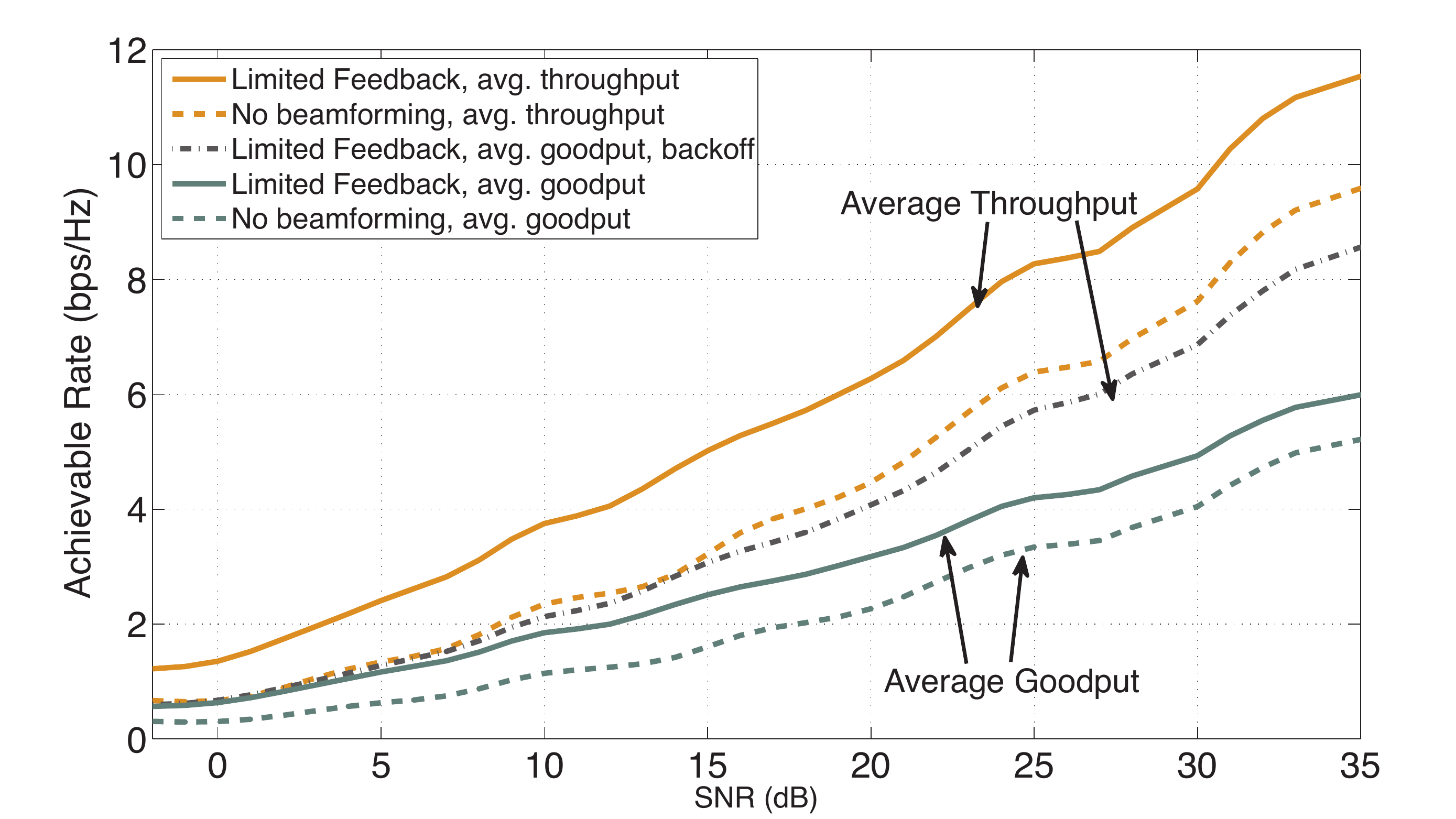}
    \caption{The average rate as a function of SNR with and without rate backoff for  a limited feedback system with $N_f = N_b = 4$, and $N_{of} = 95$ femtocells per cell site, and $\rmB = 5$ bits on the feedback channel.}
    \vspace{-5pt}
    \label{fig:ratebackoff}
  \end{center}
\end{figure}

\begin{figure}[t]
  \begin{center}
    \includegraphics[scale = 0.6]{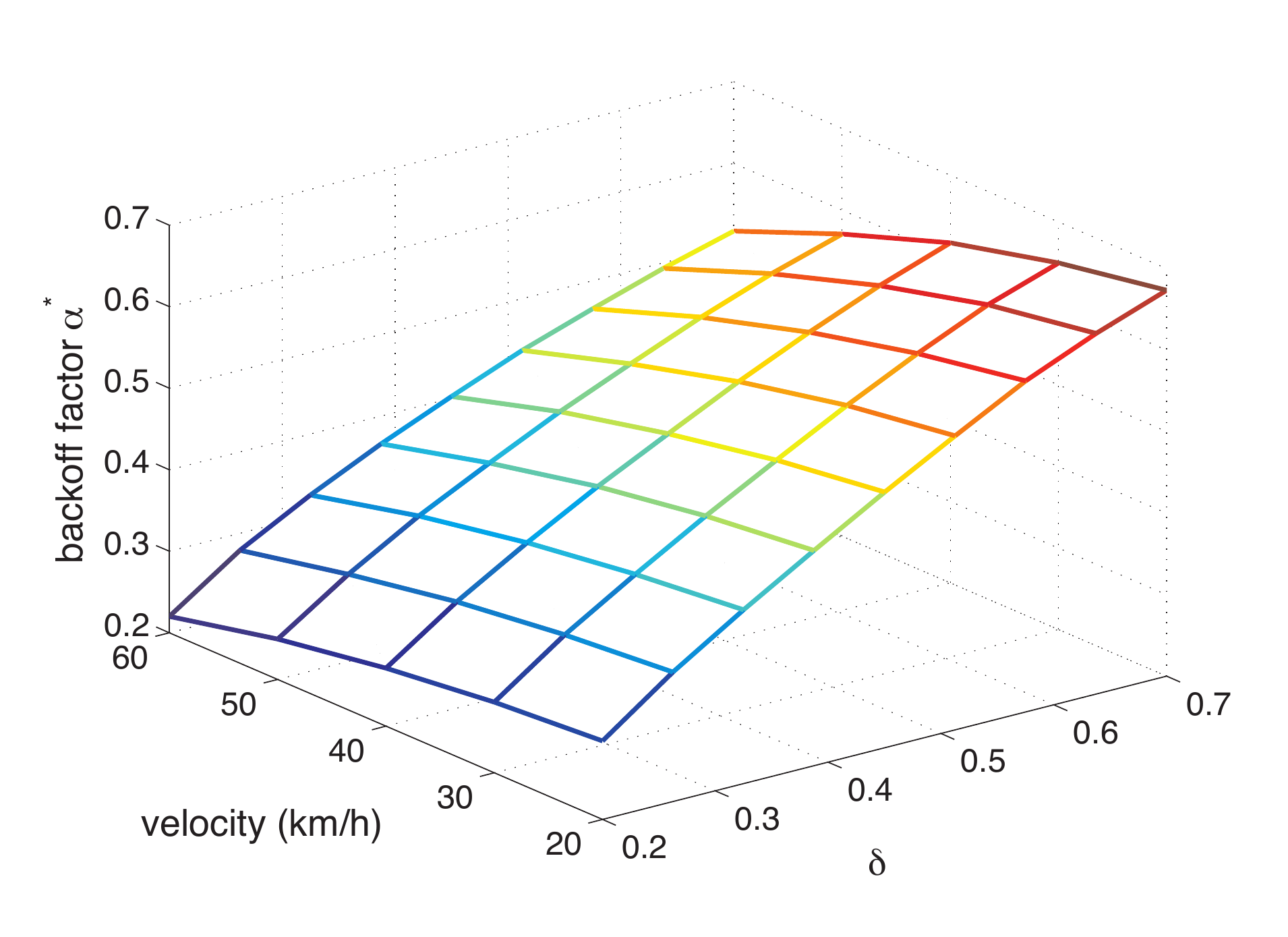}
    \caption{The backoff factor $\beta^*$ as a function of delta $\delta \in [0.2\; 0.8]$, and velocity $v \in [20\; 60] $ km/h, averaged over $100$ users uniformly distributed inside the macrocell, with $N_{of} = 95$ femtocells per cell site. }
    \label{fig:alphadeltavel}
  \end{center}
\end{figure}

\end{document}